\documentclass[journal]{IEEEtran}
\usepackage{amssymb}
\usepackage{amsmath}
\usepackage{amsfonts}
\newtheorem{theorem}{Theorem}

\newtheorem{definition}{Definition}
\newtheorem{lemma}{Lemma}

\newtheorem{assumption}[definition]{Assumption}

\newcommand{\cancel}[1]{}

\newcommand{\bool}{\{0,1\}}

\def\F{\mathbb{F}}

\def \bf{\textbf}

\def \mc{\mathcal}
\def \mb{\mathbf}
\def \ms{\mathsf}

\def \mcA{\mathcal{A}}

\def \negl{\ms{negl}}

\def \pk{\ms{pk}}
\def \sk{\ms{sk}}

\def \lar{\leftarrow}

\begin{document}

\title{A CCA2 Secure Variant of the McEliece Cryptosystem}

\author{Nico D\"{o}ttling, Rafael~Dowsley, J\"{o}rn~M\"{u}ller-Quade and Anderson~C.~A.~Nascimento

\thanks{Rafael~Dowsley is with the Department of Computer Science and Engineering, University of California at San Diego (UCSD), 9500 Gilman Drive, La Jolla, California 92093, USA. Email: rdowsley@cs.ucsd.edu. This work was partially done while the author was with the Department of Electrical Engineering, University of Brasilia. Supported in part by NSF grant CCF-0915675 and by a Focht-Powell fellowship.}
\thanks{Nico D\"{o}ttling and J\"{o}rn~M\"{u}ller-Quade are with the Institute Cryptography and Security, Karlsruhe Institute of Technology. Am Fasanengarten 5, 76128 Karlsruhe, Germany. E-mail: \{ndoett,muellerq\}@ira.uka.de}
\thanks{Anderson~C.~A.~Nascimento is with the Department of Electrical Engineering, University of Brasilia. Campus Universit\'{a}rio Darcy Ribeiro, Bras\'{i}lia, CEP: 70910-900, Brazil. E-mail: andclay@ene.unb.br.}
\thanks{A preliminary version of this work, enciphering just a single message rather than many possibly correlated ones, has appeared at the proceedings of  CT-RSA -- 2009~\cite{DMN09}.}
}

\markboth{}%
{}
\maketitle

\begin{abstract}
The McEliece public-key encryption scheme has become an interesting alternative to cryptosystems based on number-theoretical problems. Differently from RSA and ElGamal, McEliece PKC is not known to be broken by a quantum computer. Moreover, even tough McEliece PKC has a relatively big key size, encryption and decryption operations are rather efficient. 
In spite of all the recent results in coding theory based cryptosystems, to the date, there are no constructions secure against chosen ciphertext attacks in the standard model -- the \emph{de facto} security notion for public-key cryptosystems. 

In this work, we show the first construction of a McEliece based public-key cryptosystem secure against chosen ciphertext attacks in the standard model. Our construction is inspired by a recently proposed technique by Rosen and Segev.
\end{abstract}

\IEEEpeerreviewmaketitle

\section{Introduction}

Indistinguishability of messages under adaptive chosen ciphertext attacks is one of the strongest known notions of security for public-key encryption schemes (PKE). Many computational assumptions have been used in the literature for obtaining cryptosystems meeting such a strong security notion. Given one-way trapdoor permutations, we know how to obtain CCA2 security from any semantically secure public-key cryptosystem \cite{NY89,Sah99,Lin03}. Efficient constructions are also known based on number-theoretic assumptions \cite{CS98} or on identity based encryption schemes \cite{CHK04}. Obtaining a CCA2 secure cryptosystem (even an inefficient one) based on the McEliece assumptions in the standard model has been an open problem in this area for quite a while. We note, however, that secure schemes in the random oracle model have been proposed in \cite{KI01}.

Recently, Rosen and Segev proposed an elegant and simple new computational assumption for obtaining CCA2 secure PKEs: \emph{correlated products} \cite{RS09}. They provided constructions of correlated products based on the existence of certain \emph{lossy trapdoor functions} \cite{PW08} which in turn can be based on the decisional Diffie-Hellman problem and on Paillier's decisional residuosity problem \cite{PW08}.

In this paper, we show that ideas similar to those of Rosen and Segev can be applied for obtaining an efficient construction of a CCA2 secure PKE built upon the McEliece assumption. Inspired by the definition of correlated products \cite{RS09}, we define a new kind of PKE called $k$-repetition CPA secure cryptosystem and provide an adaptation of the construction proposed in \cite{RS09} to this new scenario. Such cryptosystems can be constructed from very weak (one-way CPA secure) PKEs and randomized encoding functions. In contrast, Rosen and Segev give a more general, however less efficient, construction of correlated secure trapdoor functions from lossy trapdoor functions. We show directly that a randomized version of the McEliece cryptosystem \cite{NIKM07} is $k$-repetition CPA secure and obtain a CCA2 secure scheme in the standard model. The resulting cryptosystem encrypts many bits as opposed to the single-bit PKE obtained in \cite{RS09}. We expand the public and secret-keys and the ciphertext by a factor of $k$ when  compared to the original McEliece PKE. 

In a concurrent and independent work~\cite{GV08}, Goldwasser and Vaikuntanathan proposed a new CCA2 secure public-key encryption scheme based on lattices using the construction by Rosen and Segev. Their scheme assumed that the problem of
learning with errors (LWE) is hard~\cite{Reg05}.

A direct construction of correlated products based on  McEliece and Niederreiter  PKEs has been obtained by Persichetti~\cite{Per12} in a subsequent work.

\section{Preliminaries}

\subsection{Notation}

If $x$ is a string, then $|x|$ denotes its length, while $|S|$ represents the cardinality of a set $S$. If $n \in \mathbb{N}$ then $1^n$ denotes the string of $n$ ones. $s \leftarrow S$ denotes the operation of choosing an element $s$ of a set $S$ uniformly at random. $w \leftarrow \mathcal{A}(x,y,\ldots)$ represents the act of running the algorithm $\mathcal{A}$ with inputs $x, y, \ldots$ and producing output $w$. We write $w \leftarrow \mathcal{A}^{\mathcal{O}}(x,y,\ldots)$ for representing an algorithm $\mathcal{A}$ having access to an oracle $\mathcal{O}$. We denote by Pr[$E$] the probability that the event $E$ occurs. If $a$ and $b$ are two strings of bits or two matrices, we denote by $a|b$ their concatenation. The transpose of a matrix $M$ is $M^T$. If $a$ and $b$ are two strings of bits, we denote by $\langle a,b \rangle$ their dot product modulo $2$ and by $a \oplus b$ their bitwise XOR. $\mathcal{U}_n$ is an oracle that returns an uniformly random element of $\bool^n$. 

We use the notion of randomized encoding-function for functions $\ms{E}$ that take an input $\ms{m}$ and random coins $\ms{s}$ and output a randomized representation $\ms{E}(\ms{m};\ms{s})$ from which $\ms{m}$ can be recovered using a decoding-function $\ms{D}$. We will use such randomized encoding-functions to make messages entropic or unguessable.

\subsection{Public-Key Encryption Schemes}

A Public-Key Encryption Scheme ($\mathsf{PKE}$) is defined as follows:

\begin{definition}(Public-Key Encryption). A public-key encryption scheme is a triplet of algorithms $\mathsf{(Gen}$, $\mathsf{Enc}$, $\mathsf{Dec)}$ such that:

\begin{itemize}
\item $\mathsf{Gen}$ is a probabilistic polynomial-time key generation algorithm which takes as input a security parameter $1^n$ and outputs a public-key $\mathsf{pk}$ and a secret-key $\mathsf{sk}$. The public-key specifies the message space $\mathcal{M}$ and the ciphertext space $\mathcal{C}$.
\item $\mathsf{Enc}$ is a (possibly) probabilistic polynomial-time encryption algorithm which receives as input a public-key $\mathsf{pk}$, a message $\mathsf{m} \in \mathcal{M}$ and random coins $\ms{r}$, and outputs a ciphertext $\mathsf{c} \in \mathcal{C}$. We write $\ms{Enc}(\ms{pk},\ms{m};\ms{r})$ to indicate explicitly that the random coins $\ms{r}$ are used and $\ms{Enc}(\ms{pk},\ms{m})$ if fresh random coins are used.
\item $\mathsf{Dec}$ is a deterministic polynomial-time decryption algorithm which takes as input a secret-key $\mathsf{sk}$ and a ciphertext $\mathsf{c}$, and outputs either a message $\mathsf{m} \in \mathcal{M}$ or an error symbol $\perp$.
\item (Completeness) For any pair of public and secret-keys generated by $\mathsf{Gen}$ and any message $\mathsf{m} \in \mathcal{M}$ it holds that $\mathsf{Dec}(\mathsf{sk},\mathsf{Enc}(\mathsf{pk},\mathsf{m};\ms{r}))=\mathsf{m}$ with overwhelming probability over the randomness used by $\mathsf{Gen}$ and the random coins $\ms{r}$ used by $\mathsf{Enc}$.
\end{itemize}
\end{definition}

A basic security notion for public-key encryption schemes is One-Wayness under chosen-plaintext attacks (OW-CPA). This notion states that every PPT-adversary $\mc{A}$, given a public-key $\ms{pk}$ and a ciphertext $\ms{c}$ of a uniformly chosen message $\ms{m} \in \mc{M}$, has only negligible probability of recovering the message $\ms{m}$ (The probability runs over the random coins used to generate the public and secret-keys, the choice of $\ms{m}$ and the coins of $\mc{A}$).

Below we define the standard security notions for public-key encryption schemes, namely, indistinguishability against chosen-plaintext attacks (IND-CPA)~\cite{GM84} and against adaptive chosen-ciphertext attacks (IND-CCA2)~\cite{RS91}. Our game definition follows the approach of~\cite{HK07}.

\begin{definition} (IND-CPA security). To a two-stage adversary $\mathcal{A} = (\mathcal{A}_1, \mathcal{A}_2)$ against $\mathsf{PKE}$ we associate the following experiment.

\smallskip
\framebox{
\begin{minipage}[t]{3.0in}
$\mathsf{Exp}^{cpa}_{\mathsf{PKE},\mathcal{A}}(n)$:

$(\mathsf{pk},\mathsf{sk}) \leftarrow \mathsf{Gen}(1^n)$

$(\mathsf{m}^0, \mathsf{m}^1, state) \leftarrow \mathcal{A}_1(\mathsf{pk})$ s.t. $|\mathsf{m}^0|=|\mathsf{m}^1|$

$b \leftarrow \bool$

$\mathsf{c}^* \leftarrow \mathsf{Enc}(\mathsf{pk},\mathsf{m}^b)$

$b' \leftarrow \mathcal{A}_2(\mathsf{c}^*, state)$

If $b=b'$ return $1$, else return $0$.
\end{minipage}}\\

We define the advantage of $\mathcal{A}$ in the experiment as

\[\mathsf{Adv}^{cpa}_{\mathsf{PKE},\mathcal{A}}(n)=\left|Pr\left[
\mathsf{Exp}^{cpa}_{\mathsf{PKE},\mathcal{A}}(n)=1\right]-\frac{1}{2}\right|
\]

We say that $\mathsf{PKE}$ is indistinguishable against chosen-plaintext attacks (IND-CPA) if for all probabilistic polynomial-time (PPT) adversaries $\mathcal{A} = (\mathcal{A}_1, \mathcal{A}_2)$ the advantage of $\mathcal{A}$ in the above experiment is a negligible function of $n$.

\end{definition}

\begin{definition} (IND-CCA2 security). To a two-stage adversary $\mathcal{A} = (\mathcal{A}_1, \mathcal{A}_2)$ against $\mathsf{PKE}$ we associate the following experiment.

\smallskip
\framebox{
\begin{minipage}[t]{3.0in}
$\mathsf{Exp}^{cca2}_{\mathsf{PKE},\mathcal{A}}(n)$:

$(\mathsf{pk},\mathsf{sk}) \leftarrow \mathsf{Gen}(1^n)$

$(\mathsf{m}^0, \mathsf{m}^1, state) \leftarrow \mathcal{A}_1^{\mathsf{Dec}(\mathsf{sk},\cdotp)}(\mathsf{pk})$ s.t. $|\mathsf{m}^0|=|\mathsf{m}^1|$

$b \leftarrow \bool$

$\mathsf{c}^* \leftarrow \mathsf{Enc}(\mathsf{pk},\mathsf{m}^b)$

$b' \leftarrow \mathcal{A}_2^{\mathsf{Dec}(\mathsf{sk},\cdotp)}(\mathsf{c}^*, state)$

If $b=b'$ return $1$, else return $0$.
\end{minipage}}\\

The adversary $\mathcal{A}_2$ is not allowed to query $\mathsf{Dec}(\mathsf{sk},\cdotp)$ with $\mathsf{c}^*$. We define the advantage of $\mathcal{A}$ in the experiment as

\[\mathsf{Adv}^{cca2}_{\mathsf{PKE},\mathcal{A}}(n)=\left|Pr\left[
\mathsf{Exp}^{cca2}_{\mathsf{PKE},\mathcal{A}}(n)=1\right]-\frac{1}{2}\right|
\]

We say that $\mathsf{PKE}$ is indistinguishable against adaptive chosen-ciphertext attacks (IND-CCA2) if for all probabilistic polynomial-time (PPT) adversaries $\mathcal{A} = (\mathcal{A}_1, \mathcal{A}_2)$ that make a polynomial number of oracle queries the advantage of $\mathcal{A}$ in the experiment is a negligible function of $n$.

\end{definition}

\subsection{McEliece Cryptosystem}

In this Section we define the basic McEliece cryptosystem~\cite{McE78}, following \cite{Sen10} and \cite{NIKM07}. Let $\mc{F}_{n,t}$ be a family of binary linear error-correcting codes given by two parameters $n$ and $t$. Each code $C \in \mc{F}_{n,t}$ has code length $n$ and minimum distance greater than $2t$. We further assume that there exists an efficient probabilistic algorithm $\ms{Generate}_{n,t}$ that samples a code $C \in \mc{F}_{n,t}$ represented by a generator-matrix $\mb{G}_C$ of dimensions $l \times n$ together with an efficient decoding procedure $\ms{Decode}_C$ that can correct up to $t$ errors. 

The McEliece PKE consists of a triplet of probabilistic algorithms $(\mathsf{Gen}_\mathrm{McE},$ $\mathsf{Enc}_\mathrm{McE}, \mathsf{Dec}_\mathrm{McE})$ such that:

\begin{itemize}
\item The probabilistic polynomial-time key generation algorithm $\mathsf{Gen}_\mathrm{McE}$, computes $(\mb{G}_C,\ms{Decode}_C) \lar \ms{Generate}_{n,t}()$, sets $\pk = \mb{G}_C$ and $\sk = \ms{Decode}_C$ and outputs $(\pk,\sk)$.

\item The probabilistic polynomial-time encryption algorithm $\mathsf{Enc}_\mathrm{McE}$, takes the public-key $\pk = \mb{G}_C$ and a plaintext $\mathsf{m}\in \F_2^l$ as input and outputs a ciphertext $\mathsf{c}=\ms{m}\mb{G}_C\oplus \ms{e}$, where $\ms{e}\in\bool^n$ is a random vector of Hamming-weight $t$.

\item The deterministic polynomial-time decryption algorithm $\mathsf{Dec}_\mathrm{McE}$, takes the secret-key $\sk = \ms{Decode}_C$ and a ciphertext $\ms{c} \in \F_2^n$, computes $\ms{m} = \ms{Decode}_C(\ms{c})$ and outputs $\ms{m}$.
\end{itemize}

This basic variant of the McEliece cryptosystem is OW-CPA secure (for a proof see \cite{Sen10} Proposition 3.1), given that matrices $\mb{G}_C$ generated by $\ms{Generate}_{n,t}$ are pseudorandom (Assumption \ref{as:ind} below) and decoding random linear codes is hard when the noise vector has hamming weight $t$.

There exist several optimization for the basic scheme, mainly improving the size of the public-key. Biswas and Sendrier \cite{BS08} show that the public generator-matrix $\mb{G}$ can be reduced to row echelon form, reducing the size of the public-key from $l\cdot n$ to $l \cdot (n - l)$ bits. However, we cannot adopt this optimization into our scheme of section \ref{sec:rme}\footnote{Neither is it possible for the scheme of \cite{NIKM07}, on which our $k$-repetition McEliece scheme is based upon.}, as it implies a simple attack compromising IND-CPA security\footnote{The scheme of \cite{NIKM07} encrypts by computing $\ms{c} = (\ms{m} | \ms{s} )\cdot \mb{G} \oplus \ms{e})$. If $\mb{G}$ is in row-echelon form, $\ms{m} \oplus \ms{e}^\prime$ is a prefix of $\ms{c}$, where $\ms{e}^\prime$ is a prefix of $\ms{e}$. Thus an IND-CPA adversary can distinguish between the encryptions of two plaintexts $\ms{m}_0$ and $\ms{m}_1$ by checking whether the prefix of $\ms{c}^\ast$ is closer to $\ms{m}_0$ or $\ms{m}_1$.} (whereas \cite{BS08} prove OW-CPA security).

In this work we use a slightly modified version of the basic McEliece PKE scheme. Instead of sampling an error vector $\ms{e}$ by choosing it randomly from the set of vectors with Hamming-weight $t$, we generate $\ms{e}$ by choosing each of its bits according to the Bernoulli distribution $\mc{B}_{\theta}$ with parameter $\theta=\frac{t}{n}-\epsilon$ for some $\epsilon > 0$. Clearly, a simple argument based on the Chernoff bound gives us that the resulting error vector should be within the error capabilities of the code but for a negligible probability in $n$. The reason for using this error-distribution is that one of our proofs utilizes the fact that the concatenation $\ms{e}_1 | \ms{e}_2$ of two Bernoulli-distributed vectors $\ms{e}_1$ and $\ms{e}_2$ is again Bernoulli distributed. Clearly, it is not the case that $\ms{e}_1 | \ms{e}_2$ is a uniformly chosen vector of Hamming-weight $2t$ if each $\ms{e}_1$ and $\ms{e}_2$ are uniformly chosen with Hamming-weight $t$.

Using the Bernoulli error-distribution, we base the security of our scheme on the pseudorandomness of the McEliece matrices $\mb{G}$ and the pseudorandomness of the learning parity with noise (LPN) problem (see below).

\subsection{McEliece Assumptions and Attacks}

In this subsection, we discuss the hardness assumptions for the McEliece cryptosystem. Let $\mc{F}_{n,t}$ be a family of codes together with a generation-algorithm $\ms{Generate}_{n,t}$ as above and let $\mb{G}_C$ be the corresponding generator-matrices. An adversary can attack the McEliece cryptosystem in two ways: either he can try to discover the underlying structure which would allow him to decode efficiently or he can try to run a generic decoding algorithm. This high-level intuition that there are two different ways of attacking the cryptosystem can be formalized~\cite{Sen10}. Accordingly, the security of the cryptosystem is based on two security assumptions. 

The first assumption states that for certain families $\mc{F}_{n,t}$, the distribution of generator-matrices $\mb{G}_C$ output by $\ms{Generate}_{n,t}$ is pseudorandom. Let $l$ be the dimension of the codes in $\mc{F}_{n,t}$.

\begin{assumption}\label{as:ind}
Let $\mb{G}_C$ be distributed by $(\mb{G}_C,\ms{Decode}_C) \lar \ms{Generate}_{n,t}()$ and $\mb{R}$ be distributed by $\mb{R} \lar \mc{U}(\F_2^{k \times n})$. For every PPT algorithm $\mc{A}$ it holds that 
\[
| \Pr[\mc{A}(\mb{G}_C) = 1 ] - \Pr[\mc{A}(\mb{R}) = 1 ] | < \negl(n).
\]
\end{assumption}

In the classical instantiation of the McEliece cryptosystem, $\mc{F}_{n,t}$ is chosen to be the family of irreducible binary Goppa-codes of length $n = 2^m$ and dimension $l = n - tm$. For this instantiation, an efficient distinguisher was built for the case of high-rate codes~\cite{FGOPT11, FOPT10} (i.e., codes where the rate are very close to 1). But, for codes that do not have a high-rate, no generalization of the previous distinguisher is known and
the best known attacks~ \cite{CFS01,LS01} are based on the {\em support splitting algorithm}~\cite{Sen00} and have exponential runtime. Therefore, one should be careful when choosing the parameters of the Goppa-codes, but for encryption schemes it is possible to use codes that do not have high-rate.


The second security assumption is the difficulty of the \emph{decoding problem} (a classical problem in coding theory), or equivalently, the difficulty of the \emph{learning parity with noise} (LPN) problem (a classical problem in learning theory). The best known algorithms for decoding a random linear code are based on the \emph{information set decoding} technique~\cite{LB88,Leo88,Ste89}. Over the years, there have been improvements in the running time~\cite{CC98, BLP08,FS09,BLP11,MMT11,BJMM12}, but the best algorithms still run in exponential time.

Below we give the definition of LPN problem following the description of~\cite{NIKM07}.

\begin{definition}\label{def:slpn}(LPN search problem). Let $s$ be a random binary string of length $l$. We consider the Bernoulli distribution $\mathcal{B}_{\theta}$ with parameter $\theta \in (0,\frac{1}{2})$. Let $\mathcal{Q}_{s,\theta}$ be the following distribution:

\begin{center}
$\{(a,\langle s,a\rangle \oplus e)|a \leftarrow \bool^l,e \leftarrow \mathcal{B}_{\theta}\}$
\end{center}

For an adversary $\mathcal{A}$ trying to discover the random string $s$, we define its advantage as:

\begin{center}
$\mathsf{Adv}_{\mathsf{LPN}_{\theta},\mathcal{A}}(l)= \mathrm{Pr[} \mathcal{A}^{\mathcal{Q}_{s,\theta}} = s|s \leftarrow \bool^l \mathrm{]}$
\end{center}

The $\mathsf{LPN}_{\theta}$ problem with parameter $\theta$ is hard if the advantage of all PPT adversaries $\mathcal{A}$ that make a polynomial number of oracle queries is negligible.
\end{definition}

Katz and Shin \cite{KS06} introduce a distinguishing variant of the $\ms{LPN}$-problem, which is more useful in the context of encryption schemes.

\begin{definition}(LPNDP, LPN distinguishing problem). Let $s,a$ be binary strings of length $l$. Let further $\mathcal{Q}_{s,\theta}$ be as in Definition \ref{def:slpn}. Let $\mcA$ be a PPT-adversary. The distinguishing-advantage of $\mcA$ between $\mathcal{Q}_{s,\theta}$ and the uniform distribution $\mc{U}_{l+1}$ is defined as
\begin{align*}
&\ms{Adv}_{\ms{LPNDP_\theta},\mcA}(l) = \\
&\left| \mathrm{Pr} \left[ \mathcal{A}^{\mathcal{Q}_{s,\theta}}=1|s \leftarrow \bool^{l} \right]- \mathrm{Pr} \left[ \mathcal{A}^{\mathcal{U}_{l+1}}=1\right] \right|
\end{align*}
The $\mathsf{LPNDP}_{\theta}$ with parameter $\theta$ is hard if the advantage of all PPT adversaries $\mathcal{A}$ is negligible.
\end{definition}

Further, \cite{KS06} show that the $\ms{LPN}$-distinguishing problem is as hard as the $\ms{LPN}$ search-problem with similar parameters.

\begin{lemma}(\cite{KS06})\label{lem_LPN}
Say there exists an algorithm $\mathcal{A}$ making $q$ oracle queries, running in time $t$, and such that
\[
\ms{Adv}_{\ms{LPNDP}_\theta,\mcA}(l) \geq \delta
\]
Then there exists an adversary $\mathcal{A}'$ making $q'=O(q\delta^{-2}\mathrm{log}l)$ oracle queries, running in time $t'=O(t l\delta^{-2}\mathrm{log}l)$, and such that
\[
\mathsf{Adv}_{\mathsf{LPN}_{\theta},\mathcal{A}'}(l) \geq \frac{\delta}{4}
\]
\end{lemma}

The reader should be aware that in the current state of the art, the average-case hardness of these two assumptions, as well as all other assumptions used in public-key cryptography, cannot be reduced to the worst-case hardness of a NP-hard problem\footnote{Quite remarkably, some lattice problems enjoy average-case to worst-case reductions, but these are not for problems known to be NP-hard.} (and even if that was the case, we do not even know if $\mathcal{P} \neq \mathcal{NP}$). The confidence on the hardness of solving all these problems on average-case (that is what cryptography really needs) comes from the lack of efficient solutions despite the efforts of the scientific community over the years. But more studies are, of course, necessary in order to better assess the difficulties of such problems. We should highlight that when compared to cryptosystems based on number-theoretical assumptions such as the hardness of factoring or of computing the discrete-log, the cryptosystems based on coding and lattice assumptions have the advantage that no efficient quantum algorithm breaking the assumptions is known. One should also be careful when implementing the McEliece cryptosystem as to avoid side-channel attacks \cite{ST2008}.

\begin{figure}
\begin{tabular} {c|c|c|c}
(m,t)&plaintext size&ciphertext size&security (key)\\
\hline
(10,50)&524&1024&491\\
(11,32)&1696&2048&344\\
(12,40)&3616&4096&471
\end{tabular}
\caption{A table of McEliece key parameters and security estimates taken from \cite{Sen10}.}
\end{figure}


\subsection{Signature Schemes}

Now we define signature schemes ($\mathsf{SS}$) and the security notion called one-time strong unforgeability.

\begin{definition}(Signature Scheme). A signature scheme is a triplet of algorithms $\mathsf{(Gen}$, $\mathsf{Sign}$, $\mathsf{Ver)}$ such that:

\begin{itemize}
\item $\mathsf{Gen}$ is a probabilistic polynomial-time key generation algorithm which takes as input a security parameter $1^n$ and outputs a verification key $\mathsf{vk}$ and a signing key $\mathsf{dsk}$. The verification key specifies the message space $\mathcal{M}$ and the signature space $\mathcal{S}$.
\item $\mathsf{Sign}$ is a (possibly) probabilistic polynomial-time signing algorithm which receives as input a signing key $\mathsf{dsk}$ and a message $\mathsf{m} \in \mathcal{M}$, and outputs a signature $\mathsf{\sigma} \in \mathcal{S}$.
\item $\mathsf{Ver}$ is a deterministic polynomial-time verification algorithm which takes as input a verification key $\mathsf{vk}$, a message $\mathsf{m} \in \mathcal{M}$ and a signature $\mathsf{\sigma} \in \mathcal{S}$, and outputs a bit indicating whether $\mathsf{\sigma}$ is a valid signature for $\mathsf{m}$ or not (i.e., the algorithm outputs $1$ if it is a valid signature and outputs $0$ otherwise).
\item (Completeness) For any pair of signing and verification keys generated by $\mathsf{Gen}$ and any message $\mathsf{m} \in \mathcal{M}$ it holds that $\mathsf{Ver}(\mathsf{vk}, \mathsf{m}, \mathsf{Sign}(\mathsf{dsk},\mathsf{m}))=1$ with overwhelming probability over the randomness used by $\mathsf{Gen}$ and $\mathsf{Sign}$.
\end{itemize}
\end{definition}

\begin{definition} (One-Time Strong Unforgeability). To a two-stage adversary $\mathcal{A} = (\mathcal{A}_1, \mathcal{A}_2)$ against $\mathsf{SS}$ we associate the following experiment.

\smallskip
\framebox{
\begin{minipage}[t]{3.0in}
$\mathsf{Exp}^{otsu}_{\mathsf{SS},\mathcal{A}}(n)$:

$(\mathsf{vk},\mathsf{dsk}) \leftarrow \mathsf{Gen}(1^n)$

$(\mathsf{m}, state) \leftarrow \mathcal{A}_1(\mathsf{vk})$

$\mathsf{\sigma} \leftarrow \mathsf{Sign}(\mathsf{dsk},\mathsf{m})$

$(\mathsf{m}^*,\mathsf{\sigma}^*)  \leftarrow \mathcal{A}_2(\mathsf{m}, \mathsf{\sigma}, state)$

If $\mathsf{Ver}(\mathsf{vk}, \mathsf{m}^*, \mathsf{\sigma}^*)=1$ and $(\mathsf{m}^*, \mathsf{\sigma}^*) \neq (\mathsf{m}, \mathsf{\sigma})$ return $1$, else return $0$
\end{minipage}}\\

We say that a signature scheme $\mathsf{SS}$ is one-time strongly unforgeable  if for all probabilist polynomial-time (PPT) adversaries $\mathcal{A} = (\mathcal{A}_1, \mathcal{A}_2)$ the probability that $\mathsf{Exp}^{otsu}_{\mathsf{SS},\mathcal{A}}(n)$ outputs $1$ is a negligible function of $n$. One-way functions are sufficient to construct existentially unforgeable one-time signature schemes \cite{La79,NY89}.

\end{definition}

\section{$k$-repetition PKE}

\subsection{Definitions}

We now define a $k$-repetition Public-Key Encryption.

\begin{definition}($k$-repetition Public-Key Encryption). For a $\mathsf{PKE}$ $(\mathsf{Gen}$, $\mathsf{Enc}$, $\mathsf{Dec})$ and a randomized encoding-function $\ms{E}$ with a decoding-function $\ms{D}$, we define the $k$-repetition public-key encryption scheme $(\mathsf{PKE}_k)$ as the triplet of algorithms $(\mathsf{Gen}_k$, $\mathsf{Enc}_k$, $\mathsf{Dec}_k)$ such that:

\begin{itemize}
\item $\mathsf{Gen}_k$ is a probabilistic polynomial-time key generation algorithm which takes as input a security parameter $1^n$ and calls $\mathsf{PKE}$'s key generation algorithm $k$ times obtaining the public-keys $(\mathsf{pk}_1,\ldots, \mathsf{pk}_k)$ and the secret-keys $(\mathsf{sk}_1,\ldots, \mathsf{sk}_k)$. $\mathsf{Gen}_k$ sets the public-key as $\mathsf{pk}=(\mathsf{pk}_1,\ldots, \mathsf{pk}_k)$ and the secret-key as $\mathsf{sk} =(\mathsf{sk}_1,\ldots, \mathsf{sk}_k)$.

\item $\mathsf{Enc}_k$ is a probabilistic polynomial-time encryption algorithm which receives as input a public-key $\mathsf{pk}=(\mathsf{pk}_1,\ldots, \mathsf{pk}_k)$, a message $\mathsf{m} \in \mathcal{M}$ and coins $s$ and $r_1,\ldots,r_k$, and outputs a ciphertext $\mathsf{c}= (\mathsf{c}_1,\ldots, \mathsf{c}_k)= (\mathsf{Enc}(\mathsf{pk}_1,\ms{E}(\mathsf{m};s);r_1),\ldots, \mathsf{Enc}(\mathsf{pk}_k,\ms{E}(\mathsf{m};s);r_k))$.

\item $\mathsf{Dec}_k$ is a deterministic polynomial-time decryption algorithm which takes as input a secret-key $\mathsf{sk} =(\mathsf{sk}_1,\ldots, \mathsf{sk}_k)$ and a ciphertext $\mathsf{c} = (\mathsf{c}_1,\ldots, \mathsf{c}_k)$. It outputs a message $\mathsf{m}$ if $\ms{D}(\mathsf{Dec}(\mathsf{sk}_1,\mathsf{c}_1)), \ldots, \ms{D}(\mathsf{Dec}(\mathsf{sk}_k,\mathsf{c}_k))$ are all equal to some $\mathsf{m} \in \mathcal{M}$. Otherwise, it outputs an error symbol $\perp$.

\item (Completeness) For any $k$ pairs of public and secret-keys generated by $\mathsf{Gen}_k$ and any message $\mathsf{m} \in \mathcal{M}$ it holds that $\mathsf{Dec}_k(\mathsf{sk},\mathsf{Enc}_k(\mathsf{pk},\mathsf{m}))=\mathsf{m}$ with overwhelming probability over the random coins used by $\mathsf{Gen}_k$ and $\mathsf{Enc}_k$.
\end{itemize}
\end{definition}

We also define security properties that the $k$-repetition Public-Key Encryption scheme used in the next sections should meet.

\begin{definition} (Security under uniform $k$-repetition of encryption schemes). 
We say that $\mathsf{PKE}_k$ (built from an encryption scheme $\mathsf{PKE}$) is secure under uniform $k$-repetition if $\mathsf{PKE}_k$ is IND-CPA secure.
\end{definition}

\begin{definition} (Verification under uniform $k$-repetition of encryption schemes). 
We say that $\ms{PKE}_k$ is verifiable under uniform $k$-repetition if there exists an efficient deterministic algorithm $\ms{Verify}$ such that given a ciphertext $\mathsf{c} \in \mathcal{C}$, the public-key $\mathsf{pk}=(\mathsf{pk}_1,\ldots, \mathsf{pk}_k)$ and any $\mathsf{sk}_i$ for $i \in \{1, \ldots, k\}$, it holds that if $\ms{Verify}(\ms{c},\ms{pk},\ms{sk}_i) = 1$ then $\ms{Dec}_k(\ms{sk},\ms{c}) = \ms{m}$ for some $\ms{m} \neq \perp$ (i.e. $\ms{c}$ decrypts to a valid plaintext).
\end{definition}

Notice that for the scheme $\ms{PKE}_k$ to be verifiable, the underlying scheme $\ms{PKE}$ cannot be IND-CPA secure, as the verification algorithm of $\ms{PKE}_k$ implies an efficient IND-CPA adversary against $\ms{PKE}$. Thus, we may only require that $\ms{PKE}$ is OW-CPA secure.

\subsection{IND-CCA2 Security from verifiable IND-CPA Secure $k$-repetition PKE}\label{sec_cca}

In this subsection we construct the IND-CCA2 secure public-key encryption scheme $(\mathsf{PKE}_{cca2})$ and prove its security. We assume the existence of an one-time strongly unforgeable signature scheme $\ms{SS} = (\ms{Gen},\ms{Sign},\ms{Ver})$ and of a $\mathsf{PKE}_k$ that is secure and verifiable under uniform $k$-repetition. 

We use the following notation for derived keys: For a public-key $\ms{pk} = (\ms{pk}_1^0,\ms{pk}_1^1,\ldots,\ms{pk}_k^0,\ms{pk}_k^1)$ and a $k$-bit string $\ms{vk}$ we write $\ms{pk}^{\ms{vk}} = (\ms{pk}_1^{\ms{vk}_1},\ldots,\ms{pk}_k^{\ms{vk}_k})$. We will use the same notation for secret-keys $\ms{sk}$.

\begin{itemize}
\item Key Generation: $\ms{Gen}_{cca2}$ is a probabilistic polynomial-time key generation algorithm which takes as input a security parameter $1^n$. $\ms{Gen}_{cca2}$ calls $\mathsf{PKE}$'s key generation algorithm $2k$ times to obtain public-keys $\mathsf{pk}_1^0,\mathsf{pk}_1^1,\ldots, \mathsf{pk}_k^0,\mathsf{pk}_k^1$ and secret-keys $\mathsf{sk}_1^0, \mathsf{sk}_1^1, \ldots, \mathsf{sk}_k^0,\mathsf{sk}_k^1$. It sets $\ms{pk} = (\ms{pk}_1^0,\ms{pk}_1^1,\dots,\ms{pk}_k^0,\ms{pk}_k^1)$, $\ms{sk} = (\ms{sk}_1^0,\ms{sk}_1^1,\dots,\ms{sk}_k^0,\ms{sk}_k^1)$ and outputs $\ms{pk},\ms{sk})$		

\item Encryption: $\mathsf{Enc}_{cca2}$ is a probabilistic polynomial-time encryption algorithm which receives as input the public-key $\ms{pk}=(\mathsf{pk}_1^0, \mathsf{pk}_1^1 ,\ldots, \mathsf{pk}_k^0, \mathsf{pk}_k^1)$ and a message $\mathsf{m} \in \mathcal{M}$ and proceeds as follows:

\begin{enumerate}
\item Executes the key generation algorithm of the signature scheme obtaining a signing key $\mathsf{dsk}$ and a verification key $\mathsf{vk}$.
\item Compute $\mathsf{c}^\prime=\mathsf{Enc}_k(\mathsf{pk}^{\ms{vk}},\ms{m};r)$ where $r$ are random coins.
\item Computes the signature $\mathsf{\sigma}= \mathsf{Sign}(\mathsf{dsk},\ms{c}^\prime)$.
\item Outputs the ciphertext $\mathsf{c}=(\mathsf{c}^\prime, \mathsf{vk}, \mathsf{\sigma})$.
\end{enumerate}

\item Decryption:  $\mathsf{Dec}_{cca2}$ is a deterministic polynomial-time decryption algorithm which takes as input a secret-key  $\mathsf{sk} =(\mathsf{sk}_1^0,\mathsf{sk}_1^1,\ldots, \mathsf{sk}_k^0,\mathsf{sk}_k^1)$ and a ciphertext $\mathsf{c}=(\ms{c}^\prime, \mathsf{vk}, \mathsf{\sigma})$ and proceeds as follows:

\begin{enumerate}
\item If $\mathsf{Ver}(\mathsf{vk},\ms{c}^\prime, \mathsf{\sigma})=0$, it outputs $\perp$ and halts.
\item It computes and outputs $\mathsf{m}=\mathsf{Dec}_k(\ms{sk}^{\ms{vk}}, \mathsf{c}^\prime)$.
\end{enumerate}
\end{itemize}

Note that if $\mathsf{c}^\prime$ is an invalid ciphertext (i.e. not all $\ms{c}_i^\prime$ decrypt to the same plaintext), then $\mathsf{Dec}_{cca2}$ outputs $\bot$ as $\mathsf{Dec}_k$ outputs $\bot$.

As in~\cite{RS09}, we can apply a universal one-way hash function to the verification keys (as in~\cite{DDN00}) and use $k=n^{\epsilon}$ for a constant $0 < \epsilon < 1$. Note that the hash function in question need not be modeled as a random oracle. For ease of presentation, we do not apply this method in our scheme description.

\begin{theorem}\label{the_cca}
Given that $\mathsf{SS}$ is an one-time strongly unforgeable signature scheme and that $\mathsf{PKE}_k$ is IND-CPA secure and verifiable under uniform $k$-repetition, the public-key encryption scheme $\mathsf{PKE}_{cca2}$ is IND-CCA2 secure.
\end{theorem}

\begin{proof}
In this proof, we closely follow~\cite{RS09}. Denote by $\mathcal{A}$ the IND-CCA2 adversary. Consider the following sequence of games.
\begin{itemize}
	\item \bf{Game 1} This is the IND-CCA2 game.
	\item \bf{Game 2} Same as game 1, except that the signature-keys $(\ms{vk}^\ast,\ms{dsk}^\ast)$ that are used for the challenge-ciphertext $\ms{c}^\ast$ are generated before the interaction with $\mc{A}$ starts. Further, game 2 always outputs $\perp$ if $\mc{A}$ sends a decryption query $\ms{c} = (\ms{c}^\prime,\ms{vk},\sigma)$ with $\ms{vk} = \ms{vk}^\ast$.
\end{itemize}
We will now establish the remaining steps in two lemmata.

\begin{lemma}\label{lem_g1g2}
It holds that $\ms{view}_{\ms{Game 1}}(\mc{A}) \approx_c \ms{view}_{\ms{Game 2}}(\mc{A})$, given that $(\ms{Gen},\ms{Sign},\ms{Ver})$ is an one-time strongly unforgeable signature scheme.
\end{lemma}

\begin{proof}
Given that $\mc{A}$ does not send a valid decryption query $\ms{c} = (\ms{c}^\prime,\ms{vk},\sigma)$ with $\ms{vk} = \ms{vk}^\ast$ and $\ms{c} \neq \ms{c}^\ast$, $\mc{A}$'s views in game 1 and game 2 are identical. Thus, in order to distinguish game 1 and game 2 $\mc{A}$ must send a valid decryption query $\ms{c} = (\ms{c}^\prime,\ms{vk},\sigma)$ with $\ms{vk} = \ms{vk}^\ast$ and $\ms{c} \neq \ms{c}^\ast$. We will use $\mc{A}$ to construct an adversary $\mc{B}$ against the one-time strong unforgeability of the signature scheme $(\ms{Gen},\ms{Sign},\ms{Ver})$. $\mc{B}$ basically simulates the interaction of game 2 with $\mc{A}$, however, instead of generating $\ms{vk}^\ast$ itself, it uses the $\ms{vk}^\ast$ obtained from the one-time strong unforgeability experiment. Furthermore, $\ms{B}$ generates the signature $\sigma$ for the challenge-ciphertext $\ms{c}^\ast$ by using its signing oracle provided by the one-time strong unforgeability game. Whenever $\mc{A}$ sends a valid decryption query $\ms{c} = (\ms{c}^\prime,\ms{vk},\sigma)$ with $\ms{vk} = \ms{vk}^\ast$ and $\ms{c} \neq \ms{c}^\ast$, $\mc{B}$ terminates and outputs $(\ms{c}^\prime,\sigma)$. Obviously, $\mc{A}$'s output is identically distributed in Game 2 and $\mc{B}$'s simulation. Therefore, if $\mc{A}$ distinguishes between game 1 and game 2 with non-negligible advantage $\epsilon$, then $\mc{B}$'s probability of forging a signature is also $\epsilon$, thus breaking the one-time strong unforgeability of $(\ms{Gen},\ms{Sign},\ms{Ver})$.
\end{proof}

\begin{lemma}\label{lem_g2}
It holds that $\ms{Adv}_{\ms{Game 2}}(\mc{A})$ is negligible in the security parameter, given that $\ms{PKE}_k$ is verifiable and IND-CPA secure under uniform k-repetition.
\end{lemma}

\begin{proof}
Assume that $\ms{Adv}_{\ms{Game 2}}(\mc{A}) \geq \epsilon$ for some non-negligible $\epsilon$. We will now construct an IND-CPA adversary $\mc{B}$ against $\ms{PKE}_k$ that breaks the IND-CPA security of $\ms{PKE}_k$ with advantage $\epsilon$. Instead of generating $\ms{pk}$ like game 2, $\mc{B}$ proceeds as follows. Let $\ms{pk}^\ast = (\ms{pk}_1^\ast,\dots,\ms{pk}_k^\ast)$ be the public-key provided by the IND-CPA experiment to $\mc{B}$. $\mc{B}$ first generates a pair of keys for the signature scheme $(\ms{vk}^\ast,\ms{dsk}^\ast) \lar \ms{Gen}(1^n)$. Then, the public-key $\ms{pk}$ is formed by setting $\ms{pk}^{\ms{vk}^\ast} = \ms{pk}^\ast$. All remaining components $\ms{pk}_i^j$ of $\ms{pk}$ are generated by $(\ms{pk}_i^j,\ms{sk}_i^j) \lar \ms{Gen}(1^n)$, for which $\mc{B}$ stores the corresponding $\ms{sk}_i^j$. Clearly, the $\ms{pk}$ generated by $\mc{B}$ is identically distributed to the $\ms{pk}$ generated by game 2, as the $\ms{Gen}$-algorithm of $\ms{PKE}_k$ generates the components of $\ms{pk}$ independently. Now, whenever $\mc{A}$ sends a decryption query $\ms{c} = (\ms{c}^\prime,\ms{vk},\sigma)$, where $\ms{vk} \neq \ms{vk}^\ast$ (decryption queries with $\ms{vk} = \ms{vk}^\ast$ are not answered by game 2), $\mc{B}$ picks an index $i$ with $\ms{vk}_i \neq \ms{vk}^\ast_i$ and checks if $\ms{Verify}(\ms{c}^\prime,\ms{pk},\ms{sk}_i^{\ms{vk}_i}) = 1$, if not it outputs $\perp$. Otherwise it computes $\ms{m} = \ms{D}(\ms{Dec}(\ms{sk}_i,\ms{c}^\prime_i))$. Verifiability guarantees that it holds that $\ms{Dec}_k(\ms{sk}^{\ms{vk}},\ms{c}^\prime) = \ms{m}$, i.e. the output $\ms{m}$ is identically distributed as in game 2. When $\mc{A}$ sends the challenge-messages $\ms{m}_0,\ms{m}_1$, $\mc{B}$ forwards $\ms{m}_0,\ms{m}_1$ to the IND-CPA experiments and receives a challenge-ciphertext $\ms{c}^{\ast \prime}$. $\mc{B}$ then computes $\sigma = \ms{Sign}(\ms{dsk}^\ast,\ms{c}^{\ast \prime})$ and sends $\ms{c}^\ast = (\ms{c}^{\ast \prime},\ms{vk}^\ast,\sigma)$ to $\mc{A}$. This $\ms{c}^\ast$ is identically distributed as in game 2. Once $\mc{A}$ produces output, $\mc{B}$ outputs whatever $\mc{A}$ outputs. Putting it all together, $\mc{A}$'s views are identically distributed in game 2 and in the simulation of $\mc{B}$. Therefore it holds that $\ms{Adv}_{\ms{IND-CPA}}(\mc{B}) = \ms{Adv}_{\ms{Game 2}}(\mc{A}) \geq \epsilon$. Thus $\mc{B}$ breaks the IND-CPA security of $\ms{PKE}_k$ with non-negligible advantage $\epsilon$, contradicting the assumption.
\end{proof}

Plugging Lemma \ref{lem_g1g2} and Lemma \ref{lem_g2} together immediately establishes that any PPT IND-CCA2 adversary $\mc{A}$ has at most negligible advantage in winning the IND-CCA2 experiment for the scheme $\ms{PKE_{cca2}}$.
\end{proof}

\section{A Verifiable $k$-repetition McEliece Scheme}\label{sec:rme}

In this section, we will instantiate a verifiable $k$-repetition encryption scheme $\ms{PKE}_{\ms{McE,k}} = (\ms{Gen}_{\ms{McE,k}},\ms{Enc}_{\ms{McE,k}},\ms{Dec}_{\ms{McE,k}})$ based on the McEliece cryptosystem. 

In~\cite{NIKM07} it was proved that the cryptosystem obtained by changing the encryption algorithm of the McEliece cryptosystem to encrypt $\mathsf{s|m}$ (where $\mathsf{s}$ is random padding) instead of just encrypting the message $\ms{m}$ (the so called Randomized McEliece cryptosystem) is IND-CPA secure, if $|s|$ is chosen sufficiently large for the LPNDP to be hard (e.g. linear in the security-parameter $n$). We will therefore use the randomized encoding-function $\ms{E}(\ms{m};\ms{s}) = \ms{s} | \ms{m}$ (with $|s| \in \Omega(n)$) in our verifiable $k$-repetition McEliece scheme. As basis scheme $\ms{PKE}$ for our verifiable $k$-repetition McEliece scheme we use the OW-CPA secure textbook McEliece with a Bernoulli error-distribution.

The verification algorithm $\ms{Verifiy}_{\ms{McE}}(\ms{c},\ms{pk},\ms{sk}_i)$ works as follows. Given a secret-key $\ms{sk}_i$ from the secret-key vector $\ms{sk}$, it first decrypts the $i$-th component of $\ms{c}$ by $\ms{x} = \ms{Dec}_{\ms{McE}}(\ms{sk}_i,\ms{c}_i)$. Then, for all $j = 1,\dots,k$, it checks whether the vectors $\ms{c}_j \oplus \ms{x} \mb{G}_j$ have a Hamming-weight smaller than $t$, where $\mb{G}_j$ is the generator-matrix given in $\ms{pk}_j$. If so, $\ms{Verify}_{\ms{McE}}$ outputs $1$, otherwise $0$. Clearly, if $\ms{Verify}_{\ms{McE}}$ accepts, then all ciphertexts $\ms{c}_j$ are close enough to the respective codewords $\ms{x} \ms{G}_j$, i.e. invoking $\ms{Dec}_{\ms{McE}}(\ms{sk}_j,\ms{c}_j)$ would also output $\ms{x}$. Therefore, we have that $\ms{Verifiy}_{\ms{McE}}(\ms{c},\ms{pk},\ms{sk}_i) = 1$, if and only if $\ms{Dec}_{\ms{McE,k}}(\ms{sk},\ms{c}) = \ms{m}$ for some $\ms{m} \in \mc{M}$.

\subsection{Security of the k-repetition Randomized McEliece}

We now prove that the modified Randomized McEliece is IND-CPA secure under $k$-repetition.

By the completeness of each instance, the probability that in one instance $i \in \{1,\ldots,k\}$ a correctly generated ciphertext is incorrectly decoded is negligible. Since $k$ is polynomial, it follows by the union bound that the probability that a correctly generated ciphertext of $\mathsf{PKE}_{k,McE}$ is incorrectly decoded is also negligible. So $\mathsf{PKE}_{k,McE}$ meets the completeness requirement.

Denote by $\mathbf{R}_1,\ldots, \mathbf{R}_k$ random matrices of size $l \times n$, by $\mathbf{G}_1,\ldots, \mathbf{G}_k$ the public-key matrices of the McEliece cryptosystem and by $\mathsf{e}_1,\ldots, \mathsf{e}_k$ the error vectors. Define $l_1=|\mathsf{s}|$ and $l_2=|\mathsf{m}|$. Let $\mathbf{R}_{i,1}$ and $\mathbf{R}_{i,2}$ be the $l_1 \times n$ and $l_2 \times n$ sub-matrices of $\mathbf{R}_i$ such that $\mathbf{R}_i^T=\mathbf{R}_{i,1}^T|\mathbf{R}_{i,2}^T$. Define $\mathbf{G}_{i,1}$ and $\mathbf{G}_{i,2}$ similarly.

\begin{lemma}\label{lem_ind}
The scheme $\ms{PKE}_{\ms{McE},k}$ is IND-CPA secure, given that both the McEliece assumption and the LPNDP assumption hold.
\end{lemma}
\begin{proof}
Let $\mc{A}$ be an IND-CPA adversary against $\ms{PKE}_{\ms{McE},k}$. Consider the following three games.

\begin{itemize}
	\item \bf{Game 1} This is the IND-CPA game.
	\item \bf{Game 2} Same as game 1, except that the components $\ms{pk}_i$ of the public-key $\ms{pk}$ are computed by $\ms{pk}_i=(\mb{R}_i, t, \mc{M},\mc{C})$ instead of $\ms{pk}_i=(\mb{G}_i, t, \mc{M},\mc{C})$, where $\mb{R}_i$ is a randomly chosen matrix of the same size as $\mb{G}_i$
	\item \bf{Game 3} Same as game 2, except that the components $\ms{c}_i$ of the challenge-ciphertext $\ms{c}^\ast$ are not computed by $\ms{c}_i= (\ms{s} | \ms{m})\mathbf{R}_i\oplus \ms{e}_i$ but rather chosen uniformly at random.
\end{itemize}

Indistinguishability of game 1 and game 2 follows by a simple hybrid-argument using the McEliece assumption, we omit this for the sake of brevity. The indistinguishability of game 2 and game 3 can be established as follows. First observe that it holds that $\ms{c}_i= (\ms{s} | \ms{m})\mathbf{R}_i\oplus \ms{e}_i = (\ms{s} \mb{R}_{i,1} \oplus \ms{e}_i) \oplus \ms{m} \mb{R}_{i,2}$ for $i = 1,\dots,k$. Setting $\mb{R}_1 = \mb{R}_{1,1} | \dots, | \mb{R}_{k,1}$, $\mb{R}_2 = \mb{R}_{1,2} | \dots, | \mb{R}_{k,2}$ and $\ms{e} = \ms{e}_1 | \dots | \ms{e}_k$, we can write $\ms{c}^\ast = (\ms{s} \mb{R}_1 \oplus \ms{e}) \oplus \ms{m} \mb{R}_2$. Now, the LPNDP assumption allows us to substitute $\ms{s} \mb{R}_1 \oplus \ms{e}$ with a uniformly random distributed vector $\ms{u}$, as $\ms{s}$ and $\mb{R}_1$ are uniformly distributed and $\ms{e}$ is Bernoulli distributed. Therefore $\ms{c}^\ast = \ms{u} \oplus \ms{m} \mb{R}_2$ is also uniformly distributed. Thus we have reached game 3. $\mc{A}$'s advantage in game 3 is obviously 0, as the challenge-ciphertext $\ms{c}^\ast$ is statistically independent of the challenge bit $\ms{b}$. This concludes the proof.
\end{proof}

\section{Generalized Scheme}

As in~\cite{RS09}, it is possible to generalize the scheme to encrypt correlated messages instead of encrypting $k$ times the same message $m$. In this Section, we show that a similar approach is possible for our scheme, yielding an IND-CCA2 secure McEliece variant that has asymptotically the same ciphertext expansion as the efficient IND-CPA scheme of \cite{KI01}. We now present a generalized version of our encryption scheme using a correlated plaintext space.

\subsection{Definitions}

\begin{definition}($\tau$-Correlated Messages)
We call a tuple of messages $(\mathsf{m}_1,\ldots, \mathsf{m}_k)$ $\tau$-correlated for some constant $0 < \gamma < 1$ and $\tau=(1-\gamma)k$, if given any $\tau$ messages of tuple it is possible to efficiently recover all the messages. We denote the space of such messages tuples by $\mathcal{M}_{\ms{Cor}}$.
\end{definition}

Basically, $\tau$-correlated messages can be erasure-corrected. Now we define a correlated public-key encryption scheme.

\begin{definition}(Correlated Public-Key Encryption). For a $\ms{PKE}$ $(\ms{Gen}$, $\ms{Enc}$, $\ms{Dec})$ and a randomized encoding-function $\ms{E}$ that maps from the plaintext-space $\mc{M}$ to the correlated plaintext-space $\mc{M}_{\ms{Cor}}$ (with corresponding decoding-function $\ms{D}$), we define the correlated public-key encryption scheme $(\ms{PKE_{Cor}})$ as the triplet of algorithms $(\ms{Gen_{Cor}}$, $\ms{Enc_{Cor}}$, $\ms{Dec_{Cor}})$ such that:

\begin{itemize}
\item $\ms{Gen_{Cor}}$ is a probabilistic polynomial-time key generation algorithm which takes as input a security parameter $1^n$ and calls $\mathsf{PKE}$'s key generation algorithm $k$ times obtaining the public-keys $(\mathsf{pk}_1,\ldots, \mathsf{pk}_k)$ and the secret-keys $(\ms{sk}_1,\ldots, \ms{sk}_k)$. $\ms{Gen_{Cor}}$ sets the public-key as $\ms{pk}=(\ms{pk}_1,\ldots, \ms{pk}_k)$ and the secret-key as $\ms{sk} =(\ms{sk}_1,\ldots, \ms{sk}_k)$.

\item $\ms{Enc_{Cor}}$ is a probabilistic polynomial-time encryption algorithm which receives as input a public-key $\ms{pk}=(\ms{pk}_1,\ldots, \ms{pk}_k)$ and a message $\ms{m} \in \mathcal{M}$. The algorithm computes $\tilde{\ms{m}} = (\tilde{\ms{m}}_1,\dots,\tilde{\ms{m}}_k) = \ms{E}(\ms{m};s)$ (with fresh random coins $s$) and outputs the ciphertext $\ms{c}= (\ms{c}_1,\ldots, \ms{c}_k)= (\ms{Enc}(\ms{pk}_1,\tilde{\ms{m}}_1), \ldots, \ms{Enc}(\ms{pk}_k,\tilde{\ms{m}}_k))$.

\item $\ms{Dec_{Cor}}$ is a deterministic polynomial-time decryption algorithm which takes as input a secret-key $\ms{sk} =(\ms{sk}_1,\ldots, \ms{sk}_k)$ and a ciphertext $\ms{c}= (\ms{c}_1,\ldots, \ms{c}_k)$. It first computes a tuple  $\tilde{\ms{m}}=(\tilde{\ms{m}}_1,\ldots, \tilde{\ms{m}}_k) \in \mc{M}_{\ms{Cor}}$, outputs $\ms{m} = \ms{D}(\tilde{\ms{m}})$ if $\tilde{\ms{m}} \in \mc{M}_{\ms{Cor}}$, if not it outputs an error symbol $\perp$.

\item (Completeness) For any $k$ pairs of public and secret-keys generated by $\ms{Gen_{Cor}}$ and any message $\ms{m}=(\ms{m}_1,\ldots, \ms{m}_k) \in \mc{M}_{\ms{Cor}}$ it holds that $\ms{Dec_{Cor}}(\ms{sk},\ms{Enc_{Cor}}(\ms{pk},\ms{m}))=\ms{m}$ with overwhelming probability over the randomness used by $\ms{Gen_{Cor}}$ and $\ms{Enc_{Cor}}$.
\end{itemize}
\end{definition}

We also define security properties that the Correlated Public-Key Encryption scheme used in the next sections should meet.

\begin{definition} (Security of Correlated Public-Key Encryption). 
We say that $\mathsf{PKE_{Cor}}$ (built from an encryption scheme $\mathsf{PKE}$) is secure if $\mathsf{PKE_{Cor}}$ is IND-CPA secure.
\end{definition}

\begin{definition} ($\tau$-Verification). 
We say that $\ms{PKE_{Cor}}$ is $\tau$-verifiable if the exists a efficient deterministic algorithm $\ms{Verify}$, such that given a ciphertext $\ms{c} \in \mc{C}$, the public-key $\ms{pk}=(\ms{pk}_1,\ldots, \ms{pk}_k)$ and any $\tau$ distinct secret-keys $\ms{sk}_T = (\ms{sk}_{t_1}, \ldots, \ms{sk}_{t_\tau})$ (with $T = \{t_1,\dots,t_\tau \}$), it holds that if $\ms{Verify}(\ms{c},\ms{pk},T,\ms{sk}_T) = 1$ then $\ms{Dec_{Cor}}(\ms{sk},\ms{c}) = \ms{m}$ for some $\ms{m} \neq \perp$ (i.e. $\ms{c}$ decrypts to a valid plaintext).
\end{definition}

\subsection{IND-CCA2 Security from IND-CPA Secure Correlated PKE}\label{sec_ccacor}

We now describe the IND-CCA2 secure public-key encryption scheme $(\mathsf{PKE'}_{cca2})$ built using the correlated PKE and prove its security. We assume the existence of 
a correlated PKE, $\mathsf{PKE_{Cor}}$, that is secure and $\tau$-verifiable. We also use an error correcting code $\mathsf{ECC}: \Sigma^l \rightarrow \Sigma^k$ with minimum distance $\tau$ and polynomial-time encoding. Finally, we assume the existence of an one-time strongly unforgeable signature scheme $\ms{SS} = (\ms{Gen},\ms{Sign},\ms{Ver})$ in which the verification keys are elements of $\Sigma^l$ (we assumed that the verification keys are elements of $\Sigma^l$ only for simplicity, we can use any signature scheme if there is a injective mapping from the set of verification keys to $\Sigma^l$). 

We will use the following notation: For a codeword $\ms{d} = (\ms{d}_1,\dots,\ms{d}_k) \in \ms{ECC}$, set $\ms{pk}^{\ms{d}} = (\ms{pk}_1^{\ms{d}_1},\dots,\ms{pk}_k^{\ms{d}_k})$. Analogously for $\ms{sk}$.

\begin{itemize}
\item Key Generation: $\ms{Gen^\prime}_{cca2}$ is a probabilistic polynomial-time key generation algorithm which takes as input a security parameter $1^n$. $\ms{Gen^\prime}_{cca2}$ proceeds as follows. It calls $\ms{PKE}$'s key generation algorithm $|\Sigma|k$ times obtaining the public-keys $(\ms{pk}_1^1, \ldots, \ms{pk}_1^{|\Sigma|},\ldots, \ms{pk}_k^1, \ldots, \ms{pk}_k^{|\Sigma|} )$ and the secret-keys $(\ms{sk}_1^1, \ldots, \ms{sk}_1^{|\Sigma|}, \ldots, \ms{sk}_k^1, \ldots, \ms{sk}_k^{|\Sigma|} )$. Outputs $\ms{pk} = (\ms{pk}_1^1, \ldots, \ms{pk}_1^{|\Sigma|},\ldots, \ms{pk}_k^1, \ldots, \ms{pk}_k^{|\Sigma|} )$ and $\ms{sk} = (\ms{sk}_1^1, \ldots, \ms{sk}_1^{|\Sigma|}, \ldots, \ms{sk}_k^1, \ldots, \ms{sk}_k^{|\Sigma|} )$.

\item Encryption: $\ms{Enc^\prime}_{cca2}$ is a probabilistic polynomial-time encryption algorithm which receives as input the public-key $\ms{pk}=(\ms{pk}_1^1, \ldots, \ms{pk}_1^{|\Sigma|},\ldots, \ms{pk}_k^1, \ldots, \ms{pk}_k^{|\Sigma|} )$ and a message $\ms{m}=(\ms{m}_1,\ldots, \ms{m}_k) \in \mc{M}$ and proceeds as follows:
\begin{enumerate}
\item Executes the key generation algorithm of the signature scheme $\ms{SS}$ obtaining a signing key $\ms{dsk}$ and a verification key $\ms{vk}$. Computes $\ms{d}=\ms{ECC}(\ms{vk})$. Let $\ms{d}_i$ denote the $i$-element of $\ms{d}$.
\item Computes $\ms{c}^\prime=\ms{Enc_{Cor}}(\ms{pk}^{\ms{d}},\ms{m})$.
\item Computes the signature $\sigma= \ms{Sign}(\ms{dsk},\ms{c}^\prime)$.
\item Outputs the ciphertext $\mathsf{c}=(\ms{c}^\prime, \ms{vk}, \sigma)$.
\end{enumerate}

\item Decryption: $\ms{Dec^\prime}_{cca2}$ is a deterministic polynomial-time decryption algorithm which takes as input a secret-key  $\ms{sk} = (\ms{sk}_1^1, \ldots, \ms{sk}_1^{|\Sigma|}, \ldots, \ms{sk}_k^1, \ldots, \ms{sk}_k^{|\Sigma|} )$ and a ciphertext $\ms{c}=(\ms{c}^\prime, \ms{vk}, \sigma)$ and proceeds as follows:

\begin{enumerate}
\item If $\ms{Ver}(\ms{vk},\ms{c}^\prime, \sigma)=0$, it outputs $\perp$ and halts. Otherwise, it performs the following steps.
\item Compute $\ms{d}=\ms{ECC}(\ms{vk})$.
\item Compute $\ms{m}=\ms{Dec_{Cor}}(\ms{sk}^{\ms{d}}, \ms{c})$ and output $\ms{m}$.
\end{enumerate}

\end{itemize}

\begin{theorem}\label{the_ccacor}
Given that $\mathsf{SS}$ is an one-time strongly unforgeable signature scheme and that $\mathsf{PKE_{Cor}}$ is secure and $\tau$-verifiable, the public-key encryption scheme $\mathsf{PKE'}_{cca2}$ is IND-CCA2 secure.
\end{theorem}

\begin{proof}
The proof is almost identical to the proof of theorem~\ref{the_cca}. Denote by $\mathcal{A}$ the IND-CCA2 adversary. Consider the following two of games.

\begin{itemize}
	\item \bf{Game 1} This is the IND-CCA2 game.
	\item \bf{Game 2} Same as game 1, except that the signature-keys $(\ms{vk}^\ast,\ms{dsk}^\ast)$ that are used for the challenge-ciphertext $\ms{c}^\ast$ are generated before the interaction with $\mc{A}$ starts. Further, game 2 terminates and outputs $\perp$ if $\mc{A}$ sends a decryption query with $\ms{c} = (\ms{c}^\prime,\ms{vk},\sigma)$ with $\ms{vk} = \ms{vk}^\ast$.
\end{itemize}

Again, we will split the proof of Theorem \ref{the_ccacor} in two lemmata.

\begin{lemma}\label{lem2_g1g2}
From $\mc{A}$'s view, game 1 and game 2 are computationally indistinguishable, given that $\ms{SS}$ is an existentially unforgeable one-time signature-scheme.
\end{lemma}

We omit the proof, since it is identical to the proof of lemma \ref{lem_g1g2}.

\begin{lemma}\label{lem2_g2}
It holds that $\ms{Adv}_{\ms{Game 2}}(\mc{A})$ is negligible in the security parameter, given that $\ms{PKE_{Cor}}$ is verifiable IND-CPA secure correlated public-key encryption scheme.
\end{lemma}

\begin{proof}
We proceed as in the proof of Lemma \ref{lem_g2}. Assume that $\ms{Adv}_{\ms{Game 2}}(\mc{A}) \geq \epsilon$ for some non-negligible $\epsilon$. We will now construct an IND-CPA adversary $\mc{B}$ against $\ms{PKE_{Cor}}$ that breaks the IND-CPA security of $\ms{PKE_{Cor}}$ with advantage $\epsilon$. Again, instead of generating $\ms{pk}$ like game 2, $\mc{B}$ will construct $\ms{pk}$ using the public-key $\ms{pk}^\prime$ provided by the IND-CPA experiment. Let $\ms{d} = \ms{ECC}(\ms{vk}^\ast)$. $\mc{B}$ sets $\ms{pk}^{\ms{d}} = \ms{pk}^\ast$. All remaining components $\ms{pk}_i^j$ of $\ms{pk}$ are generated by $(\ms{pk}_i^j,\ms{sk}_i^j) \lar \ms{Gen}(1^n)$, for which $\mc{B}$ stores the corresponding $\ms{sk}_i^j$. Obviously, the $\ms{pk}$ generated by $\mc{B}$ is identically distributed to the $\ms{pk}$ generated by game 2, as in both cases all components are $\ms{pk}_i^j$ are generated independently by the key-generation algorithm $\ms{Gen}$ of $\ms{PKE}$. Whenever $\mc{A}$ sends a decryption query with $\ms{vk} \neq \ms{vk}^\ast$, $\mc{B}$ does the following. Let $\ms{d} = \ms{ECC}(\ms{vk})$ and $\ms{d}^\ast = \ms{ECC}(\ms{vk}^\ast)$. Since the two codewords $\ms{d}$ and $\ms{d}^\ast$ are distinct and the code $\ms{ECC}$ has minimum-distance $\tau$, there exist a $\tau$-set of indices $T \subseteq \{1,\dots,k\}$ such that it holds for all $t \in T$ that $\ms{d}_t \neq \ms{d}_t^\ast$. Thus, the public-keys $\ms{pk}_t^{\ms{d}_t}$, for $t \in T$ were generated by $\mc{B}$ and it thus knows the corresponding secret-keys $\ms{sk}_t^{\ms{d}_t}$. $\mc{B}$ checks if $\ms{Verify}(\ms{c}^\prime,\ms{pk}^{\ms{d}},T,\ms{sk}^{\ms{d}}_T) = 1$ holds, i.e. if $\ms{c}^\prime$ is a valid ciphertext for $\ms{PKE_{Cor}}$ under the public-key $\ms{pk}^{\ms{d}}$. If so, $\mc{B}$ decrypts $\tilde{\ms{m}}_T = (\tilde{\ms{m}}_{t} | t \in T) = (\ms{Dec}(\ms{sk}_t^{\ms{d}_t},\ms{c}^\prime_t) | t \in T)$. Since the plaintext-space $\mc{M}_{\ms{Cor}}$ is $\tau$-correlated, $\mc{B}$ can efficiently recover the whole message $\tilde{\ms{m}}$ from the $\tau$-submessage $\tilde{\ms{m}}_T$. Finally, $\mc{B}$ decodes $\ms{m} = \ms{D}(\tilde{\ms{m}})$ to recover the message $\ms{m}$ and outputs $\ms{m}$ to $\mc{A}$. Observe that the verifiability-property of $\ms{PKE_{Cor}}$ holds regardless of the subset $T$ used to verify. Thus, from $\mc{A}$'s view the decryption-oracle behaves identically in game 2 and in $\mc{B}$'s simulation.

Finally, when $\mc{A}$ sends its challenge messages $\ms{m}_0$ and $\ms{m}_1$, $\mc{B}$ forwards $\ms{m}_0$ and $\ms{m}_1$ to the IND-CPA experiment for $\ms{PKE_{Cor}}$ and receives a challenge-ciphertext $\ms{c}^{\ast \prime}$. $\mc{B}$ then computes $\sigma = \ms{Sign}(\ms{sk}^\ast,\ms{c}^{\ast \prime})$ and outputs the challenge-ciphertext $\ms{c}^\prime = (\ms{c}^{\ast \prime},\ms{vk}^\ast,\sigma)$ to $\mc{A}$. When $\mc{A}$ generates an output, $\mc{B}$ outputs whatever $\mc{A}$ outputs.

Putting it all together, $\mc{A}$'S views are identically distributed in game 2 and $\mc{B}$'s simulation. Therefore, it holds that $\ms{Adv}_{\ms{IND-CPA}}(\mc{B}) = \ms{Adv}_{\ms{game 2}}(\mc{A}) \geq \epsilon$. Thus, $\mc{B}$ breaks the IND-CPA security of $\ms{PKE_{Cor}}$ with non-negligible advantage $\epsilon$, contradicting the assumption.
\end{proof}

Plugging Lemma \ref{lem2_g1g2} and Lemma \ref{lem2_g2} establish that any PPT IND-CCA2 adversary $\mc{A}$ has at most negligible advantage in winning the IND-CCA2 experiment for the scheme $\ms{PKE^\prime_{cca2}}$.
\end{proof}

\subsection{Verifiable Correlated PKE based on the McEliece Scheme}

We can use a modified version of the scheme presented in Section~\ref{sec:rme} to instantiate a $\tau$-correlated verifiable IND-CPA secure McEliece scheme $\ms{PKE}_{McE,Cor}$. A corresponding IND-CCA2 secure scheme is immediately implied by the construction in  Section~\ref{sec_ccacor}. As plaintext-space $\mc{M}_{\ms{Cor}}$ for $\ms{PKE}_{McE,Cor}$, we choose the set of all tuples $(\ms{s} | \ms{y}_1,\dots,\ms{s} | \ms{y}_k)$, where $\ms{s}$ is a $n$-bit string and $(\ms{y}_1,\dots,\ms{y}_k)$ is a codeword from code $\ms{C}$ that can efficiently correct $k - \tau$ erasures. Clearly, $\mc{M}_{\ms{Cor}}$ is $\tau$-correlated. Let $\ms{E_C}$ be the encoding-function of $\ms{C}$ and $\ms{D}_{\ms{C}}$ the decoding-function of $\ms{C}$. The randomized encoding-function $\ms{E_{McE,Cor}}$ used by $\ms{PKE_{McE,Cor}}$ proceeds as follows. Given a message $\ms{m}$ and random coins $\ms{s}$, it first computes $(\ms{y}_1,\dots,\ms{y}_k) = \ms{E_C}(\ms{m})$ and outputs $(\ms{s} | \ms{y}_1,\dots,\ms{s} | \ms{y}_k)$. The decoding-function $\ms{D_{McE,Cor}}$ takes a tuple $(\ms{s} | \ms{y}_1,\dots,\ms{s} | \ms{y}_k)$ and outputs $\ms{D_C}(\ms{y}_1,\dots,\ms{y}_k)$. Like in the scheme of Section~\ref{sec:rme}, the underlying OW-CPA secure encryption-scheme $\ms{PKE}$ is textbook-McEliece.

The $\tau$-correlatedness of $\ms{PKE_{McE,Cor}}$ follows directly by the construction of $\mc{M}_{\ms{Cor}}$, $\ms{E_{Mce, Cor}}$ and $\ms{D_{Mce, Cor}}$. It remains to show verifiability and IND-CPA security of the scheme. The $\ms{Verify_{McE}}$-algorithm takes a ciphertext $\ms{c} = (\ms{c}_1,\dots,\ms{c}_k)$, a public-key $\ms{pk}$, an a partial secret-key $\ms{sk}_T$ (for a $\tau$-sized index-set $T$) and proceeds as follows. First, it decrypts the components of $\ms{c}$ at the indices of $T$, i.e. it computes $\ms{x}_t = \ms{Dec_{McE}}(\ms{sk}_t,\ms{c}_t)$ for $t \in T$. Then, it checks whether all $\ms{x}_t$ are of the form $\ms{x}_t = \ms{s} | \ms{y}_t$ for the same string $\ms{s}$. If not, it stops and outputs 0. Next, it constructs a vector $\tilde{\ms{y}} \in \Sigma^k$ with $\tilde{\ms{y}}_i = \ms{y}_i$ for $i \in T$ and $\tilde{\ms{y}}_i = \perp$ (erasure) for $i \notin T$. $\ms{Verify}$ then runs the erasure-correction algorithm of $\ms{C}$ on $\tilde{\ms{y}}$. If the erasure-correction fails, it stops and outputs 0. Otherwise let $\ms{y} = (\ms{y}_1,\dots,\ms{y}_k)$ be the corrected vector returned by the erasure-correction. Then, $\ms{Verify}$ sets $\ms{x} = (\ms{s} | \ms{y}_1,\dots,\ms{s} | \ms{y}_k)$. Let $\mb{G}_1,\dots,\mb{G}_k$ be the generator-matrices given in $\ms{pk}_1,\dots,\ms{pk}_k$. Finally, $\ms{Verify}$ checks whether all the vectors $\ms{c}_j \oplus \ms{x} \mb{G}_j$, for $j = 1,\dots,k$, have Hamming-weight smaller than $t$. If so, it outputs 1, otherwise 0. Clearly, if $\ms{Verify_{McE}}$ outputs 1, then the ciphertext-components $\ms{c}_j$ of $\ms{c}$ are valid McEliece encryptions.

The IND-CPA-security is proven analogously to Lemma \ref{lem_ind}. First, the McEliece generator-matrices $\mb{G}_i$ are replaced by random matrices $\mb{R}_i$, then, using the LPNDP-assumption, vectors of the form $\ms{s} \ms{R}_i \oplus \ms{e}_i$ are replaced by uniformly random vectors $\ms{u}_i$. Likewise, after this transformation the adversarial advantage is 0.

\section{Acknowledgments} 
We would like to thank Edoardo Persichetti for comments on the definition of $k$-repetition PKE in an earlier version of this work. We  would like to thank the anonymous referees who provided us with valuable feedback that greatly improved the quality of this paper.

\end{document}